\documentclass[10pt]{article}

\usepackage[leqno]{amsmath}
\usepackage{amssymb}
\usepackage{amsthm}
\usepackage{enumerate}
\usepackage{verbatim}
\usepackage{mathrsfs}
\usepackage{dsfont}
\usepackage{graphicx}

\newcommand{\tu}{\textup}

\newcommand{\ts}{\textstyle}

\usepackage{soul}

\newcommand{\mr}{\mathrm}

\newcommand{\mc}{\mathcal}
\newcommand{\mb}{\mathbf}

\newcommand{\uhr}{\upharpoonright}

\newcommand{\tsfrac}[2]{{\ts\frac{#1}{#2}}}

\newcommand{\BAR}[1]{\overline{#1}}

\let\Ex\undefined
\let\Pr\undefined

\DeclareMathOperator*{\Ex}{\mathds{E}}
\DeclareMathOperator*{\Pr}{\mathds{P}}

\newcommand{\Bin}{\mr{Bin}}

\usepackage{mathtools}
\newcommand{\defeq}{\vcentcolon=}

\newcommand{\fieldfont}[1]{\mathbb{#1}}
\newcommand{\N}{\fieldfont{N}}

\newtheoremstyle{theorem-style}
  {}
  {}
  {\slshape}
  {}
  {\bf}
  {.}
  {.5em}
  {}

\newtheorem{thm}{Theorem}

\newtheorem{prop}[thm]{Proposition}
\newtheorem{la}[thm]{Lemma}

\newtheorem{main-la}[thm]{Main Lemma}

\newtheorem{cor}[thm]{Corollary}

\theoremstyle{definition}
\newtheorem{df}[thm]{Definition}

\usepackage{fullpage}
\usepackage{float}
\usepackage{hyperref}

\def\[#1\]{\begin{align*}#1\end{align*}}

\renewcommand{\ul}{\underline}

\theoremstyle{definition}
\newtheorem{op}[thm]{Open Problem}

\newcommand{\DT}{\mathsf{DT}}
\newcommand{\DL}{\mathsf{DL}}
\newcommand{\ODL}{\mathsf{ODL}}
\newcommand{\wODL}{\mathsf{wODL}}

\newcommand{\DTd}{\mathsf{DT_{\mr{depth}}}}
\newcommand{\DLw}{\mathsf{DL_{\mr{width}}}}
\newcommand{\DNF}{\mathsf{DNF}}
\newcommand{\CNF}{\mathsf{CNF}}
\newcommand{\DNFw}{\mathsf{DNF}_{\mr{width}}}
\newcommand{\CNFw}{\mathsf{CNF}_{\mr{width}}}

\newcommand{\U}[2]{U(#2)}

\newcommand{\width}{\mb{width}}

\renewcommand{\rho}{\varrho}

\usepackage{bm}
\renewcommand{\mb}{\bm}

\newcommand{\ACzero}{\mr{AC}^0}

\usepackage{framed}

\newcommand{\AND}{\textsc{and}}
\newcommand{\OR}{\textsc{or}}

\usepackage{mdframed}
\mdfdefinestyle{MyQuoteFrame}{%
linecolor=black,
rightline=false, topline=false, bottomline=false,
outerlinewidth=0.1pt,
innerleftmargin=1em,
leftmargin=2em
backgroundcolor=white}

\begin{document}

\title{Shrinkage of Decision Lists and DNF Formulas}
\author{Benjamin Rossman\thanks{Duke University. Email \texttt{benjamin.rossman@duke.edu}}}
\date{\today}
\maketitle{}

\begin{abstract}
We establish nearly tight bounds on the expected shrinkage of decision lists and DNF formulas under the $p$-random restriction $\mb R_p$ for all values of $p \in [0,1]$. For a function $f$ with domain $\{0,1\}^n$, let $\DL(f)$ denote the minimum size of a decision list that computes $f$. We show that
\[
  \mathds E[\ \DL(f{\uhr}\mb R_p) \ ] \le 
  \DL(f)^{\log_{2/(1-p)}(\frac{1+p}{1-p})}.
\]
For example, this bound is $\sqrt{\DL(f)}$ when $p = \sqrt{5}-2 \approx 0.24$. For Boolean functions $f$, we obtain the same shrinkage bound with respect to DNF formula size plus $1$ (i.e., replacing $\DL(\cdot)$ with $\DNF(\cdot)+1$ on both sides of the inequality).
\end{abstract}

\section{Introduction}

Random restrictions are a powerful tool in circuit complexity and the analysis of Boolean functions.  
A {\em restriction} is a partial assignment to the input bits of a function $f$ on the hypercube $\{0,1\}^n$.
For a parameter $p \in [0,1]$, the {\em $p$-random restriction} $\mb R_p$ independently leaves each input bit free with probability $p$ and otherwise assigns it to $0$ or $1$ with equal probability. 
We denote by $f{\uhr}\mb R_p$ the function obtained from $f$ by restricting its inputs to the subcube of $\{0,1\}^n$ that correspond to $\mb R_p$.

Random restrictions are known to reduce the complexity of functions in simple models of computations, such as decision trees ($\DT$), decision lists ($\DL$), DNF formulas ($\DNF$), and DeMorgan formulas ($\mc L$); the symbols in parentheses are notation for the corresponding size measures (see Section \ref{sec:prelims} for definitions).
With respect to DeMorgan formula leaf-size $\mc L$, 
it is easy to see that $\mc L(f{\uhr}\mb R_p)$ has expectation at most $p \cdot \mc L(f)$. (This follows by linearity of expectation from the observation that each input literal in a minimal formula for $f$ is eliminated by $\mb R_p$ with probability $p$.)
Subbotovskaya \cite{subbotovskaya1961realizations} was the first to show that the expected shrinkage factor is in fact significantly smaller than $p$ (she showed an upper bound $O(p^{3/2})$ for $p \ge 1/\mc L(f)^{2/3}$).
A subsequent line of results \cite{andreev1987method,impagliazzo1993effect,paterson1993shrinkage,hastad1998shrinkage,tal2014shrinkage}, culminating in an $p^{2-o(1)}$ bound of H{\aa}stad \cite{hastad1998shrinkage} and a low-order improvement by Tal \cite{tal2014shrinkage}, eventually established an asymptotically tight bound:

\begin{thm}[Shrinkage of DeMorgan formulas \cite{tal2014shrinkage}]\label{thm:DeMorgan}
For all Boolean functions $f$,
\[
  \Ex[\ \mc L(f{\uhr}\mb R_p)\ ] = O(\,p^2 \mc L(f) + p \sqrt{\mc L(f)}\,).
\]
\end{thm}

The constant $2$ 
in the exponent $p$ 
in Theorem \ref{thm:DeMorgan} 
is known as the ``shrinkage exponent'' of DeMorgan formulas.
Shrinkage under $\mb R_p$ has also been studied for restricted types of formulas, namely read-once, monotone, and bounded-depth ($\ACzero$).
It was shown in \cite{dubiner1993read,haastad1995shrinkage} that read-once formulas
have shrinkage exponent
$\log_{\sqrt 5-1}(2) \approx 3.27$. 
The shrinkage exponent of {monotone formulas} is between $2$ and $\log_{\sqrt 5-1}(2)$ and conjectured to equal the latter; determining the exact constant is a longstanding question
(Open Problem \ref{op:m}).
In the $\ACzero$ setting (bounded-depth formulas with unbounded \AND{} and \OR{} gates), it is known that \mbox{depth-$d$} formulas with fan-in $m$ shrink to expected size $O(1)$ under $\mb R_p$ when $p$ is $O(1/\log m)^{d-1}$ \cite{rossman:LIPIcs:2019:10823}. However, it is open to determine the shrinkage rate for larger $p$, particularly in the ``mild random restriction'' regime where $p$ is $\Omega(1)$ or $1-o(1)$ (Open Question \ref{op:AC0}).

The results of this paper give nearly tight bounds on the shrinkage under $\mb R_p$ of depth-$2$ formulas (also known as DNF and CNF formulas), as well as the more general computational model of {\em decision lists}. 
Before stating our main result, it is instructive to first consider shrinkage in the simpler model of {\em decision trees}.
For a function $f$ on the hypercube (with domain $\{0,1\}^n$ and arbitrary range), we denote by $\DT(f)$ the minimum number of leaves (i.e., output nodes) in a decision tree that computes $f$. The following bound is shown by straightforward induction on $\DT(f)$. (I believe this bound is probably folklore, but could not find a reference so have included the short proof in Section \ref{sec:DTshrinkage}.)

\begin{thm}[Shrinkage of decision trees]\label{thm:DTshrinkage}
For all functions $f$ on the hypercube, 
\[
  \Ex[\ \DT(f{\uhr}\mb R_p)\ ] \le \DT(f)^{\log_2(1+p)}.
\]
This bound holds with equality when $f$ is a parity function.
\end{thm}

Decision lists are a natural computational model that has been studied in many contexts \cite{blum1992rank,bshouty1996subexponential,krause2006computational,hancock1996lower,rivest1987learning}.
A {\em decision list of size $m$} is a sequence 
$L = ((C_1,b_1),\dots,(C_m,b_m))$ where $b_1,\dots,b_m$ are arbitrary output values and $C_1,\dots,C_m$ are conjunctive clauses (\AND{}s of literals) such that $C_1 \vee \dots \vee C_m$ is a tautology.\footnote{In other words, every input $x \in \{0,1\}^n$ satisfies at least one of $C_1,\dots,C_m$.
Without loss of generality, $C_m$ may be chosen as the empty (always true) conjunctive clause $\top$. We allow $C_1 \vee \dots \vee C_m$ to be an arbitrary tautology in order to more naturally define the class of {\em orthogonal} decision lists later on in Section \ref{sec:orthogonal}.}
$L$~computes a function on the hypercube as follows: on input $x \in \{0,1\}^n$, the output is $b_i$ for the first index $i \in [m]$ such that $C_i(x)$ is satisfied.
We denote by $\DL(f)$ the minimum size of a decision list that computes~$f$.

Decision lists are a generalization decision trees: every decision tree is equivalent to a decision list of the same size, and thus $\DL(f) \le \DT(f)$ for all functions $f$ on the hypercube.\footnote{The name ``decision list'' elsewhere commonly refers to (what we call) width-1 decision trees, in which each clause is a single literal (i.e., an input variable $x_i$ or its negation $\BAR{x_i}$).
Whereas unbounded-width decision lists are a generalization decision trees, width-1 decision lists are instead a special case.}
Boolean decision lists, in which $b_1,\dots,b_m \in \{0,1\}$, are moreover a generalization of both DNF and CNF formulas.  In particular, {\em DNF formulas} are the special case where $b_1=\dots=b_{m-1}=1$ and $b_m=0$. Following custom, we count the {\em size} of a DNF formula as $m-1$ instead of $m$, and thus $\DL(f) \le \DNF(f)+1$ for all Boolean functions $f$.

Despite decision lists and DNF/CNF formulas being more complex computational models than decision trees, our main result  
shows that they shrink at a similar rate under $\mb R_p$.

\begin{thm}[Shrinkage of decision lists and DNF formulas]\label{thm:DLshrinkage}
For all functions $f$ on the hypercube, 
\[
  \Ex[\ \DL(f{\uhr}\mb R_p)\ ] \le \DL(f)^{\gamma(p)}
  \quad\text{ where }\quad \gamma(p) \defeq
  \ts\log_{\frac{2}{1-p}}(\frac{1+p}{1-p}).
\]
If $f$ is Boolean, then also
$
  \Ex[\ \DNF(f{\uhr}\mb R_p) + 1\ ] \le (\DNF(f) + 1)^{\gamma(p)}
$
(and similarly for $\CNF(\cdot)+1$).
\end{thm}

\begin{figure}[H]\label{fig:gamma}
\begin{center}
\includegraphics[scale=0.25]{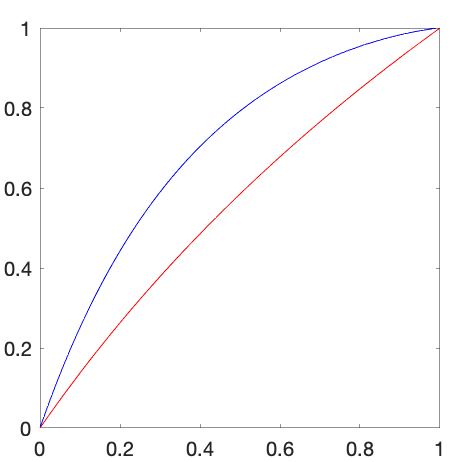}
\caption{Plots of $\gamma(p) \defeq \ts\log_{\frac{2}{1-p}}(\frac{1+p}{1-p})$ (blue) and $\log_2(1+p)$ (red)}\vspace{-.1in}
\end{center}
\end{figure}

Note that $\gamma : [0,1] \to [0,1]$ is an increasing function with $\gamma(0) = 0$ and $\gamma(1) = 1$ (see Figure \ref{fig:gamma}). 
The bound of Theorem \ref{thm:DLshrinkage} is thus nontrivial for all values of $p \in (0,1)$. 
This bound is moreover close to optimal: $\log_2(1+p)$ is a lower bound on the best possible function $\gamma(p)$ 
(Section \ref{sec:lb}).  As corollaries, we obtain additional bounds $\mathsf{ODL}(f)^{\gamma(p)}$ and $\mathsf{wODL}(f)^{\gamma(p)}$ on the shrinkage of orthogonal and weakly orthogonal decision lists (Corollary \ref{cor:orthogonal}), as well as $(\mc L_2(f)+1)^{\gamma(2p)}$ for depth-$2$ formula leaf-size (Corollary \ref{cor:depth2}).

Theorem \ref{thm:DLshrinkage} yields the following bounds for particular settings of $p$ in terms of $m = \DL(f)$:
\[
\Ex[\ \DL(f{\uhr}\mb R_p)\ ] 
\le 
\begin{cases}
2 
&\text{for }p = O({\tsfrac{1}{\log m}})
,\vphantom{\Big|}\\
\sqrt{m}
&\text{for }p = \sqrt{5}-2 \approx 0.24
,\vphantom{\Big|}\\
m/2 
&\text{for }p = 1 - O({\tsfrac{\log\log m}{\log m}})
,\vphantom{\Big|}\\
m - 1
&\text{for }p = 1 - O({\tsfrac{\log m}{m}})
.\vphantom{\Big|}
\end{cases}
\]
For small $p = O(1/\log m)$, a variant of H{\aa}stad's Switching Lemma 
(discussed below)
actually implies a stronger inequality $\Ex[\ \DT(f{\uhr}\mb R_p)\ ]$ $\le 2$ with $\DT$ in place of $\DL$ (Corollary \ref{cor:SLsize}). 
Theorem \ref{thm:DLshrinkage} is mainly interesting for larger values of $p$.
In particular, the ``mild random restriction'' regime when $p$ is $\Omega(1)$ or $1 - o(1)$ has important applications in pseudorandomness \cite{gopalan2012better,razborov2015pseudorandom}, DNF sparsification \cite{gopalan2013dnf,lovett2020decision} and hypercontractivity \cite{o2014analysis}.

\subsection{Switching lemmas and size measures vs.\ width/depth measures}\label{sec:SL}

We have so far discussed the shrinkage of various complexity measures 
under the $p$-random restriction $\mb R_p$. 
The switching lemmas stated below can be viewed as apples-to-oranges shrinkage results that bound one complexity measure on $f{\uhr}\mb R_p$ in terms of another complexity measure on $f$.
Here there is a useful distinction between ``size measures'' $\DT$, $\DL$, $\DNF$
and their corresponding ``width/depth measures'', denoted by $\DTd$, $\DLw$, $\DNFw$.
Width/depth measures are typically related to the logarithm of size measures: functions with size complexity $m$ are approximable by (or in some cases equivalent to) functions with width/depth complexity $O(\log m)$.
H{\aa}stad's Switching Lemma \cite{Hastad86} gives a tail bound on the decision tree size of $f{\uhr}\mb R_p$ in terms of the decision list width of $f$.\footnote{In its application to $\ACzero$ circuit lower bounds, Theorem \ref{thm:SL} is usually stated (more narrowly) in the form 
\[
\Pr[\ \CNFw(f{\uhr}\mb R_p) \ge t\ ] \le O(p\cdot\DNFw(f))^t
\]
for Boolean functions $f$.  The name ``Switching Lemma'' refers to the conversion of a DNF formula to a CNF formula. 
The more general bound stated in Theorem \ref{thm:SL} is implicit in proofs of \cite{Hastad86}.}

\begin{thm}[Switching Lemma \cite{Hastad86}]\label{thm:SL}
For all functions $f$ on the hypercube and $t \in \N$,
\[
  \Pr[\ \DTd(f{\uhr}\mb R_p) \ge t\ ]
  \le
  O(p\cdot\DLw(f))^t.
\]
\end{thm}

A variant of the Switching Lemma with $\log \DL(f)$ in place of $\DLw(f)$ was proved in \cite{rossman:LIPIcs:2019:10823}.

\begin{thm}[Switching Lemma in terms of decision list size \cite{rossman:LIPIcs:2019:10823}]\label{thm:SLsize}
For every function $f$ on the hypercube and $t \in \N$,
\[
  \Pr[\ \DTd(f{\uhr}\mb R_p) \ge t\ ]
  \le
  O(p\cdot\log \DL(f))^t.
\]
\end{thm}

We remark that Theorem \ref{thm:SLsize} follows directly from Theorem \ref{thm:SL} for $t \le O(\log\DL(f))$ (by the standard width reduction argument), but not for larger $t$.  Obtaining a tail bound for all $t \in \N$ is essentially to the following:

\begin{cor}[Decision tree size of decision lists]\label{cor:SLsize}
For all functions $f$ on $\{0,1\}^n$,
\[
  \Ex[\ \DT(f{\uhr}\mb R_p)\ ] \le 2
  \quad\text{ and }\quad
  \DT(f) \le O(2^{(1-p)n})
  \quad\text{ where }\quad p = O(1/\log \DL(f)),\\
  \ 
\]
\end{cor}

As previously mentioned, Corollary \ref{cor:SLsize} strengthen the bound $\Ex[\ \DL(f{\uhr}\mb R_p)\ ] \le 2$ for $p = O(1/\log \DL(f))$ that follows from Theorem \ref{thm:DLshrinkage} (albeit for $p$ that is a constant factor smaller).  
However, note that Corollary \ref{cor:SLsize} is
trivial for $p$ above $\Omega(1/\log \DL(f))$.  A different switching lemma for large $p$ (even $1-o(1)$) in terms of $\DNFw(f)$ was introduced by Segerlind, Buss and Impagliazzo \cite{segerlind2004switching} and quantitatively improved by Razborov \cite{razborov2015pseudorandom}.  It is unclear if these switching lemmas for ``mild random restriction'' have analogues in terms of $\log\DL(f)$; if so, that might entail a shrinkage bound for $\DL$ that is nontrivial for all $p \in (0,1)$, although potentially weaker than Theorem \ref{thm:DLshrinkage}.

Our proof of Theorem \ref{thm:DLshrinkage} involves an application of Jensen's inequality with respect to a certain carefully defined probability distribution on the set of clauses in a decision list $L$.  This distribution is related to (but not identical to) the distribution of the first satisfied clause of $L$ under a uniform random input.
A similar convexity argument appears in the proof of Theorem \ref{thm:SLsize} in \cite{rossman:LIPIcs:2019:10823}.
A second key idea, the notion of ``useful indices'' of $L$ under a restriction $\rho$, comes from a recent paper of Lovett, Wu and Zhang \cite{lovett2020decision} who proved the following result as the main lemma in establishing tight bound on the sparsification of bounded-width decision lists.

\begin{thm}[Decision list shrinkage in terms of width \cite{lovett2020decision}]\label{thm:DLwidth}
For every function $f$ on the hypercube,
\[
  \Ex[\ \DL(f{\uhr}\mb R_p)\ ] 
  \le
  \left(\frac{4}{1-p}\right)^{\DLw(f)}.
\]
\end{thm}

Note that our main result, Theorem \ref{thm:DLshrinkage}, stands in relation to Theorem \ref{thm:DLwidth} just as Theorem \ref{thm:SLsize} does to Theorem \ref{thm:SL}: in both cases we are essentially replacing $\DLw(f)$ with $\log \DL(f)$.

\subsection{Other related work}

There are different ways to quantify the effect of random restrictions on complexity measures.  Instead of bounding expectation, one may show that shrinkage occurs with high probability.  For DeMorgan formulas, high probability shrinkage results were shown in \cite{santhanam2010fighting,komargodski2013average}.  
Shrinkage results and switching lemmas have also been studied for random restrictions other than $\mb R_p$ (see \cite{beame1994switching}). 
Very interesting recent work of Filmus, Meir and Tal \cite{FMT2021} extends the technique of H{\aa}stad \cite{hastad1998shrinkage} to obtain $p^{2-o(1)}$ factor shrinkage bounds for DeMorgan formulas under a family of pseudorandom projections that generalize $\mb R_p$.

\section{Preliminaries}\label{sec:prelims}

Throughout this paper, $p$ is an arbitrary parameter in $[0,1]$. All inequalities involving $p$ hold for all values in $[0,1]$.
We often use the special case of Jensen's inequality $\Ex[\,X^c\,] \le \Ex[\,X\,]^c$ where $X$ is a nonnegative random variable and $c \in [0,1]$ (in particular, when $c$ is $\log_2(1+p)$ or $\gamma(p)$).
We write $\N$ for the natural numbers $\{0,1,2,\dots\}$, and for $m \in \N$, we write $[m]$ for $\{1,\dots,m\}$.

\subsection{Functions and restrictions on the hypercube}

{\em Function on the hypercube} refers to any function with domain $\{0,1\}^n$ where $n$ is a positive integer.  A {\em Boolean function} is a function on the hypercube with codomain $\{0,1\}$.  (The parameter $n$ plays no role in most results in this paper, so we suppress its mention whenever possible.)

A {\em restriction} is a partial assignment of Boolean variables $x_1,\dots,x_n$ to values $0$ and $1$; this is formally defined as a function $\rho : \{1,\dots,n\} \to \{0,1,\ast\}$ where $\rho(i)=\ast$ signifies that $x_i$ is left free by $\rho$. 
We denote by $\mr{Stars}(\rho) \subseteq [n]$ the set of free variables under $\rho$.
For a function $f$ on the hypercube $\{0,1\}^n$ and a restriction $\rho$, we denote by $f{\uhr}\rho$ the restricted function on the subcube $\{0,1\}^{\mr{Stars}(\rho)}$ defined in the obvious way: $(f{\uhr}\rho)(y) = f(x)$ where $x \in \{0,1\}^n$ is the input with $x_i = y_i$ if $i \in \mr{Stars}(\rho)$ and $x_i = \rho(i)$ otherwise.

For $p \in [0,1]$, the {\em $p$-random restriction} $\mb R_p$ is the random restriction that independently leaves each variable $x_i$ free with probability $p$ and otherwise sets $x_i$ to $0$ or $1$ with equal probability.  Thus, for any particular restriction $\rho$, we have $\Pr[\ \mb R_p = \rho\ ] = p^{|\mr{Stars}(\rho)|}((1-p)/2)^{n-|\mr{Stars}(\rho)|}$.

\subsection{Complexity measures $\DL,\DT,\DNF,\CNF$ and their width/depth versions}\label{sec:compmeasures} 

\begin{df}[DNF formulas]
We first define literals, conjunctive clauses, and DNF formulas over $n$ variables.
\begin{itemize}
\item
A {\em literal} is a Boolean variable $x_i$ or negated Boolean variable $\BAR{x_i}$ where $i \in \{1,\dots,n\}$.
\item
A {\em conjunctive clause} (a.k.a.\ {\em term}) is an expression $C$ of the form $\ell_1 \wedge \dots \wedge \ell_w$ where $\ell_1,\dots,\ell_w$ are literals on disjoint variables.  The parameter $w$ is the {\em width} of $C$; this may be any nonnegative integer.  The conjunctive clause of width zero is denoted by $\top$.
\item
A {\em DNF formula} is an expression $F$ of the form $C_1 \vee \dots \vee C_m$ where $C_1,\dots,C_m$ are conjunctive clauses. The parameter $m$ is the {\em size} of $F$; this may be any nonnegative integer.  The DNF formula of size $0$ is denoted by $\bot$.  The {\em width} of $F$ is defined as the maximum width of any $C_i$.
\item
{\em CNF formulas} are defined dually (with the roles of $\vee$ and $\wedge$ exchanged).
\end{itemize}
Every literal, conjunctive clause, and DNF formula computes a Boolean function $\{0,1\}^n \to \{0,1\}$ in the usual way.
\begin{itemize}
\item
A DNF formula $F$ is a {\em tautology} if it computes the identically $1$ function. Note that any DNF formula that includes the empty conjunctive clause $\top$ is a tautology.
\end{itemize}
\end{df}

\begin{df}[Decision lists]
\ 
\begin{itemize}
\item
A {\em decision list} is an expression $L$ of the form $((C_1,b_1),\dots,(C_m,b_m))$ where $b_1,\dots,b_m$ are arbitrary output values (not necessarily Boolean) and $C_1,\dots,C_m$ are conjunctive clauses such that $C_1 \vee \dots \vee C_m$ is a tautology. 
The parameter $m$ is the {\em size} of $L$; this may be any positive integer. The {\em width} of $C$ is defined as the maximum width of any $C_i$.
\end{itemize}
A decision list $L$ computes a function $\{0,1\}^n \to \{b_1,\dots,b_m\}$ as follows: on input $x$, the output is $b_\ell$ where $i \in [m]$ is the minimum index such that $C_i(x) = 1$. (Note that the final clause $C_m$ may be replaced by $\top$ without changing the function computed by $L$.)
\end{df}

\begin{df}[Decision trees]
\
\begin{itemize}
\item
A {\em decision tree} is a rooted binary tree $T$ in which each leaf is labeled by an output value (not necessarily Boolean) and each non-leaf node is labeled by a variable $x_i$, with the edges to its two children labeled ``$x_i = 0$'' and ``$x_i = 1$''. The {\em size} of $T$ is the number of leaves; this may be any positive integer. The {\em depth} of $T$ is the maximum number of non-leaf nodes on any root-to-leaf branch; this may be any nonnegative integer.
\end{itemize}
\end{df}

\begin{df}[Associated complexity measures]
For a function $f$ with domain $\{0,1\}^n$ (and arbitrary codomain), let
\[
  \DT(f) &\defeq \tu{minimum size of a decision tree that computes }f,\\
  \DL(f) &\defeq \tu{minimum size of a decision list that computes }f,\\
\intertext{
When $f$ is Boolean, we additionally define
}
  \DNF(f) &\defeq \tu{minimum size of a DNF formula that computes }f,\\
  \CNF(f) &\defeq 
  \tu{minimum size of a CNF formula that computes }f.
\]
For constant functions $\ul 0$ and $\ul 1$, note that $\DNF(\ul 0)=0$ and $\DNF(\ul 1)=1$ according to our definition, since $\ul 0$ is computed by the empty DNF formula, while $\ul 1$ is computed by the DNF formula with a single empty clause. Also note that $\CNF(f) = \DNF(\neg f)$.

Each of the above size measures has a corresponding width/depth measure. These are denoted by
\[
  \DTd(f),\qquad
  \DLw(f),\qquad
  \DNFw(f),\qquad
  \CNFw(f).
\]
\end{df}

\begin{prop}[see \cite{blum1992rank,krause2006computational}]
These size measures satisfy the following inequalities for all Boolean functions:
\[
  1 
  \le 
  \DL 
  \le   
  \left\{\begin{array}{c}
    \DNF + 1\\
    \CNF + 1
  \end{array}\right\}
  \le 
  \DNF + \CNF 
  \le 
  \DT.
\]
The corresponding width/depth measures satisfy:
\[
  0
  \le
  \DLw 
  \le 
  \left\{\begin{array}{c}
    \DNFw\\
    \CNFw\\
    \lceil\,\log_2(\DT)\,\rceil
  \end{array}\right\}
  \le 
  \DTd 
  \le
  {\DNFw \cdot \CNFw}.
\]
The above inequalities that involve decision trees and decision lists also apply to non-Boolean functions on the hypercube.
\end{prop}

We introduce additional computational models later on: (weakly) orthogonal decision lists in Section \ref{sec:orthogonal} and $\ACzero$ formulas in Section \ref{sec:ac0formulas}.

\section{Shrinkage of decision trees and decision lists}
\label{sec:DLshrinkage}

We prove Theorems \ref{thm:DTshrinkage} and \ref{thm:DLshrinkage} in Sections \ref{sec:DTshrinkage} and \ref{sec:DLshrinkage}.  We then discuss extensions of our shrinkage bound to (weakly) orthogonal decision lists in Section \ref{sec:orthogonal} and tightness of the bounds Section \ref{sec:lb}.

\subsection{Shrinkage of decision trees}
\label{sec:DTshrinkage}

\begin{proof}[Proof of Theorem \ref{thm:DTshrinkage}]
Let $T$ be a decision tree (with arbitrary output values). We must show that
\[
  \Ex[\ \mr{size}(T{\uhr}\mb R_p)\ ] 
  &\le 
  \mr{size}(T)^{\log_2(1+p)}.
\]
We argue by induction of the size of $T$. The inequality is trivial in the base case that $T$ has size $1$.

Assume $T$ has size $m \ge 2$. Then $T$ has the form ``If $x_i = 0$ then $T_0$ else $T_1$'' where $T_0,T_1$ are decision trees of size $m_0,m_1 \ge 1$ with $m_0+m_1 = m$. Without loss of generality, $T_0$ and $T_1$ never query $x_i$. We have
\[
  \Ex[\ \mr{size}(T{\uhr}\mb R_p)\ ] 
  &=
  p\Ex\big[\ \mr{size}(T{\uhr}\mb R_p)\ \big|\ \mb R_p(x_i) = \ast\ \big]
  \\
  &\quad\,+ \frac{1-p}{2}\Big(\Ex\big[\ \mr{size}(T_0{\uhr}\mb R_p)\ \big|\ \mb R_p(x_i) = 0\ \big] + \Ex\big[\ \mr{size}(T_1{\uhr}\mb R_p)\ \big|\ \mb R_p(x_i) = 1\ \big]\Big)\hspace{-1.6in}\\
  &=
  \frac{1+p}{2}\Big(
  \Ex[\ \mr{size}(T_0{\uhr}\mb R_p)\ ]
  + \Ex[\ \mr{size}(T_1{\uhr}\mb R_p)\ ]\Big)\hspace{-1in}\\
  &\le
  \frac{1+p}{2}\Big((m_0)^{\log_2(1+p)} + (m_1)^{\log_2(1+p)}\Big)
  &&\text{(induction hypothesis)}\\
  &\le
  (1+p) \Big(\frac{m}{2}\Big)^{\log_2(1+p)}
  &&\text{(Jensen's inequality)}\\
  &=
  m^{\log_2(1+p)}.&&
\]

As for tightness of the bound: If $f$ is a parity function $f(x_1,\dots,x_k) = x_1 \oplus \dots \oplus x_k$, then we have $\DT(f) = 2^k$ and
\[  
  \Ex[\ \DT(f{\uhr}\mb R_p)\ ]
  =
  \Ex[\ 2^{\Bin(k,p)}\ ]
  &=
  \sum_{i=0}^k 2^i \Pr[\ \Bin(k,p) = i\ ]\\
  &=
  \sum_{i=0}^k \binom{k}{i} (2p)^i (1-p)^{k-i}
  =
  (1+p)^k
  =
  \DT(f)^{\log_2(1+p)}.\qedhere
\]
\end{proof}

\subsection{Shrinkage of decision lists}\label{sec:DLshrinkageproof}

We now prove our main result on the shrinkage of decision lists and DNF formulas. 

\begin{proof}[Proof of Theorem \ref{thm:DLshrinkage}]
Let $f$ be any function on the hypercube and let $p \in [0,1]$.  
(Note: Neither the hypercube dimension $n$ nor the nature of output values of $f$ play no role in our analysis.)

Let $L = ((C_1,b_1),\dots,(C_m,b_m))$ be a decision list of minimum size that computes $f$, that is, with $m = \DL(f)$.
For $\ell \in [m]$, let $|C_\ell|$ denote the width of the clause $C_\ell$ (i.e.,\ the number of literals in $C_\ell$). 
Without loss of generality, we have $|C_1|,\dots,|C_{m-1}| \ge 1$ and $|C_m|=0$ (i.e.,\ $C_m$ is the empty clause $\top$).

Following Lovett, Wu and Zhang \cite{lovett2020decision}, for a restriction $\rho$, we define the set $\U{L}{\rho} \subseteq [m]$ of {\em useful indices of $L$ under $\rho$} by
\[
  \U{L}{\rho} &\defeq \{\ell \in [m] : \exists \text{ an input $x$ consistent with $\rho$ such that $C_\ell(x) = 1$ and } C_1(x)=\dots=C_{\ell-1}(x) = 0\}.
\]
If $\U{L}{\rho} = \{\ell_1,\dots,\ell_t\}$ where $1 \le \ell_1 < \dots < \ell_t \le m$, then the restricted function $f{\uhr}\rho$ is computed by the decision list $L{\uhr}\rho$ defined by
\[
L{\uhr}\rho \defeq ((C_{\ell_1}{\uhr}\rho,b_{\ell_1}),\dots,(C_{\ell_t}{\uhr}\rho,b_{\ell_t}))
\]
where $C_{\ell_i}{\uhr}\rho$ is the sub-clause of $C_{\ell_i}$ on the variables left unrestricted by $\rho$.
(Note that $C_{\ell_1} \vee \dots \vee C_{\ell_t}$ is a tautology, so $L{\uhr}\rho$ is indeed a decision list.)
Thus, we have
\begin{equation}\label{eq:useful}
  \DL(f{\uhr}\rho) \le |\U{L}{\rho}|.
\end{equation}

\begin{mdframed}[style=MyQuoteFrame]
For example, suppose $m = 4$ and 
\[  
  C_1 = x_1 \wedge x_3,\quad 
  C_2 = \BAR{x_1} \wedge x_4,\quad 
  C_3 = 
  x_2 \wedge \BAR{x_3},\quad 
  C_4 = \top.
\]
For $\rho_1 \defeq \{x_1 \mapsto 1\}$ (the restriction fixing $x_1$ to $1$ and leaving other variables free), we have 
\[
  \U{L}{\rho_1} = \{1,3,4\},\qquad L{\uhr}\rho_1 = ((x_3,b_1),(x_2 \wedge \BAR {x_3},b_3),(\top,b_4)).
\]
For $\rho_2 \defeq \{x_1 \mapsto 1,\, x_2 \mapsto 1\}$, 
we have 
\[
  \U{L}{\rho_2} = \{1,3\},\qquad L{\uhr}\rho_2 = ((x_3,b_1),(\BAR{x_3},b_3)). 
\]
In particular, the final clause $C_4$ is not useful under $\rho_2$ (since any input consistent with $\rho_2$ satisfies $C_1$ or $C_3$).
\end{mdframed}

Now comes a key definition: let $\mu = (\mu_1,\dots,\mu_m)$ be the probability density vector (defining a probability distribution on $[m]$) 
\[
  \mu_\ell \defeq \mbox{}&\Pr_{\rho \sim \mb R_p}[\ \max(\U{L}{\rho}) = \ell \text{ and } C_\ell{\uhr}\rho \equiv 1\ ]
  \quad\text{ for } \ell \in [m-1],\\
  \mu_m \defeq \mbox{}&\Pr_{\rho \sim \mb R_p}[\ \max(\U{L}{\rho}) = m \text{ or }C_{\max(\U{L}{\rho})}{\uhr}\rho \not\equiv 1\ ].
\]
Since events $\max(\U{L}{\rho}) = \ell$ are mutually exclusive, clearly we have $\mu_1 + \dots + \mu_m = 1$. 

\begin{mdframed}[style=MyQuoteFrame] 
Note that $\max(\U{L}{\rho}) = \ell$ does not imply $C_\ell{\uhr}\rho \equiv 1$, that is, $\mu_\ell$ does not necessarily equal $\Pr_{\rho \sim \mb R_p}[\ \max(\U{L}{\rho}) = \ell\ ]$. 
This is illustrated by the restriction $\rho_2$ in the above example, for which we have $\max(\U{L}{\rho_2}) = 3$, yet $C_3{\uhr}\rho_2 = \BAR x_3 \not\equiv 1$. Restrictions $\rho_1$ and $\rho_2$ both contribute to probability mass $\mu_4$: in the case of $\rho_1$, this is because $\max(\U{L}{\rho_1}) = 4$, and in the case of $\rho_2$, this is because $C_{\max(\U{L}{\rho_2})}{\uhr}\rho_2 \not\equiv 1$.
\end{mdframed}

For each $\ell \in [m]$, we have
$\mu_\ell \le \Pr[\ C_\ell{\uhr}\rho \equiv 1\ ] = ((1-p)/2)^{|C_\ell|}$
and therefore 
\begin{equation}\label{eq:w}
  |C_\ell| \le \ts\log_{2/(1-p)}(1/\mu_\ell).
\end{equation}

We require one more definition. For a restriction $\rho$ and a useful index $\ell \in \U{L}{\rho}$, let $\rho^{(\ell)}$ be the restriction obtained by augmenting $\rho$ by the unique satisfying assignment for the clause $C_\ell$.  That is, $\rho^{(\ell)}$ fixes a variable $x_i$ to $a \in \{0,1\}$ if, and only if, $\rho$ fixes $x_i$ to $a$ or $x_i=a$ in the satisfying assignment to $C_\ell$.

As in proofs of the Switching Lemma, we will use the fact that
\begin{equation}\label{eq:stars}
  \frac{\Pr[\ \mb R_p = \rho\ ]}{\Pr[\ \mb R_p = \rho^{(\ell)}\ ]}
  =
  \left(\frac{2p}{1-p}\right)^{|\mr{Stars}(\rho) \cap \mr{Vars}(C_\ell)|}
\end{equation}
since $\rho^{(\ell)}$ has exactly $|\mr{Stars}(\rho) \cap \mr{Vars}(C_\ell)|$ fewer unrestricted variables (``stars'') than $\rho$.

As observed in \cite{lovett2020decision}, for every $\ell \in \U{L}{\rho}$, we have $\U{L}{\rho^{(\ell)}} = \U{L}{\rho} \cap [\ell]$ and therefore
\begin{equation}\label{eq:ell}
  \max(\U{L}{\rho^{(\ell)}}) = \ell 
  \quad\text{ and }\quad
  C_\ell{\uhr}\rho^{(\ell)} \equiv 1.
\end{equation}
Thus, $\rho^{(\ell)}$ contributes to the probability mass $\mu_\ell$. 

As a consequence of (\ref{eq:stars}) and (\ref{eq:ell}), we claim that for all $\ell \in [m]$,
\begin{equation}\label{eq:ellmu}
  \Pr_{\rho \sim \mb R_p}[\ \ell \in \U{L}{\rho}\ ]
  \le
  \mu_\ell \left(\frac{1+p}{1-p}\right)^{|C_\ell|} .
\end{equation}
In the case $\ell = m$, this follows from $m \in \U{L}{\rho} \Rightarrow \max(\U{L}{\rho}) = m$. For $\ell \in [m-1]$, this is shown as follows:
\[
  \Pr_{\rho \sim \mb R_p}[\ \ell \in \U{L}{\rho}\ ]
  &=
  \sum_{S \subseteq \mr{Vars}(C_\ell)}\, 
  \Pr_{\rho \sim \mb R_p}[\ \ell \in \U{L}{\rho} \text{ and } \mr{Stars}(\rho) \cap \mr{Vars}(C_\ell) = S\ ]\\  
  &\stackrel{(\ref{eq:ell})}\le
  \sum_{S \subseteq \mr{Vars}(C_\ell)}\, 
  \Pr_{\rho \sim \mb R_p}[\ \ell = \max(\U{L}{\rho^{(\ell)}}) \text{ and } C_\ell{\uhr}\rho^{(\ell)} \equiv 1 \text{ and } \mr{Stars}(\rho) \cap \mr{Vars}(C_\ell) = S\ ]\\ 
  &=
  \sum_{S \subseteq \mr{Vars}(C_\ell)}\  
  \sum_{\rho \,:\, \ell = \max(\U{L}{\rho^{(\ell)}}) \text{ and } C_\ell{\uhr}\rho^{(\ell)} \equiv 1 \text{ and } \mr{Stars}(\rho)  \cap \mr{Vars}(C_\ell) = S}\, 
  \Pr[\ \mb R_p = \rho\ ]\\ 
  &=
  \sum_{S \subseteq \mr{Vars}(C_\ell)}\ 
  \sum_{\sigma \,:\, \ell = \max(\U{L}{\sigma}) \text{ and } C_\ell{\uhr}\sigma \equiv 1}\ 
  \sum_{\rho \,:\, \rho^{(\ell)} = \sigma \text{ and } \mr{Stars}(\rho) \cap \mr{Vars}(C_\ell) = S}\, 
  \Pr[\ \mb R_p = \rho\ ]\\ 
  &\stackrel{(\ref{eq:stars})}=
  \sum_{S \subseteq \mr{Vars}(C_\ell)}\ 
  \sum_{\sigma \,:\, \ell = \max(\U{L}{\sigma}) \text{ and } C_\ell{\uhr}\sigma \equiv 1} \ 
  \sum_{\rho \,:\, \rho^{(\ell)} = \sigma \text{ and } \mr{Stars}(\rho) \cap \mr{Vars}(C_\ell) = S}\, 
  \left(\frac{2p}{1-p}\right)^{|S|} 
  \Pr[\ \mb R_p = \sigma\ ]\\ 
  &=
  \sum_{S \subseteq \mr{Vars}(C_\ell)}\,
  \left(\frac{2p}{1-p}\right)^{|S|}
  \sum_{\sigma \,:\, \ell = \max(\U{L}{\sigma}) \text{ and } C_\ell{\uhr}\sigma \equiv 1}
  \Pr[\ \mb R_p = \sigma\ ]
  \quad\text{ ($\rho$ is determined by $\sigma$ and $S$)}\\ 
  &=
  \mu_\ell \sum_{S \subseteq \mr{Vars}(C_\ell)}\,
  \left(\frac{2p}{1-p}\right)^{|S|} 
  \quad\text{ (definition of $\mu_\ell$)}\\
  &=
  \mu_\ell\left(\frac{1+p}{1-p}\right)^{|C_\ell|} 
  \quad\text{ (binomial expansion of $(1+\tsfrac{2p}{1-p})^{|C_\ell|}$)}.
\]

Finally, we obtain the shrinkage bound of Theorem \ref{thm:DLshrinkage} by the following calculation, which uses Jensen's inequality in addition to the above observations:
\[
  \Ex_{\rho \sim \mb R_p}[\ \DL(f{\uhr}\rho) \ ] 
  \stackrel{(\ref{eq:useful})}\le
  \Ex_{\rho \sim \mb R_p}[\ |\U{L}{\rho}|\ ]
  &= 
  \sum_{\ell\in[m]}\,  
  \Pr_{\rho \sim \mb R_p}[\ \ell \in \U{L}{\rho} \ ]\\
  &\stackrel{(\ref{eq:ellmu})}{=}
  \sum_{\ell\in[m]}\,  
  \mu_\ell\left(\frac{1+p}{1-p}\right)^{|C_\ell|} \\
  &\stackrel{(\ref{eq:w})}\le
  \sum_{\ell \in [m]} \mu_\ell \left(\frac{1+p}{1-p}\right)^{\log_{2/(1-p)}(1/\mu_\ell)}\\
  &=
  \Ex_{\ell \sim \mu}\bigg[\ \bigg(\frac{1}{\mu_\ell}\bigg)^{\gamma(p)}\ \bigg]
  \qquad\,\text{ (definition of $\gamma(p) = \log_{\frac{2}{1-p}}(\tsfrac{1+p}{1-p})$)}\\
  &\le
  \bigg(\Ex_{\ell \sim \mu}\bigg[\ \frac{1}{\mu_\ell}\ \bigg]\bigg)^{\gamma(p)} 
  \qquad\text{ (Jensen's inequality)}\\
  &=
  m^{\gamma(p)}.\vphantom{\Big|}
\]
Since $m = \DL(f)$, this complete the proof of our bound on decision list shrinkage.\bigskip

We shall now assume that $f$ is Boolean and $C_1 \vee \dots \vee C_m$ is a minimum size DNF formula computing $f$.  Let $L$ be the equivalent decision list $((C_1,1),\dots,(C_m,1),(\top,0))$ of size $m+1$.  The shrinkage bound
\[
  \Ex[\ \DNF(f{\uhr}\mb R_p) + 1\ ] \le (\DNF(f) + 1)^{\gamma(p)}
\]
now follows from the above analysis, noting that $\DNF(f{\uhr}\rho) + 1 \le \mr{size}(L{\uhr}\rho)$ for all restrictions $\rho$.
\end{proof}

\subsection{Shrinkage of (weakly) orthogonal decision lists}\label{sec:orthogonal}

\begin{df}
Let $L = ((C_1,b_1),\dots,(C_m,b_m))$ be a decision list. We say that $L$ is
\begin{itemize}
\item
{\em orthogonal} if each input $x$ satisfies exactly one of the conjunctive clauses $C_1,\dots,C_m$,  
\item
{\em weakly orthogonal} if 
each input $x$ satisfies at most one of $C_1,\dots,C_{m-1}$.
\end{itemize}
(Note that if $L$ is weakly orthogonal, then it remains so after replacing $C_m$ with $\top$. In contrast, an orthogonal decision list has $C_m = \top$ if and only if $m=1$.)

For a function $f$ on the hypercube, we denote by $\mathsf{(w)ODL}(f)$ the minimum size of a (weakly) orthogonal decision list that computes $f$.
These complexity measures lies in-between $\DL$ and $\DT$: 
\[
  \DL \le \mathsf{wODL} \le \mathsf{ODL} \le \DT.
\]
\end{df}

Our proof of Theorem \ref{thm:DLshrinkage} implies a shrinkage bound for $\mathsf{ODL}$ and $\mathsf{wODL}$ in the same way as for $\mathsf{DNF}+1$.  

\begin{cor}\label{cor:orthogonal}
For every function $f$ on the hypercube,
\[
  \Ex[\ \mathsf{ODL}(f{\uhr}\mb R_p)\ ] \le \mathsf{ODL}(f)^{\gamma(p)}
  \quad\text{ and }\quad
  \Ex[\ \mathsf{wODL}(f{\uhr}\mb R_p)\ ] \le \mathsf{wODL}(f)^{\gamma(p)}.
\]
\end{cor}

This follows from the observation that if $L$ is orthogonal, then so is $L{\uhr}\rho$ for any restriction $\rho$, and if $L$ is semi-orthogonal, then $L{\uhr}\rho$ is semi-orthogonal after replacing the final conjunctive clause with $\top$.

\subsection{Lower bound on the optimal $\gamma(p)$}
\label{sec:lb}

What is the optimal function $\gamma(p)$ that may be chosen
in the bound on decision list shrinkage of Theorem \ref{thm:DLshrinkage}? 
We observe that $\gamma(p)$ cannot be improved beyond $\log_2(1+p)$.
The lower bound is given by a (non-Boolean) function $f$ computed by a read-once decision tree of depth $k$ and size $2^k$, in which each internal node queries a distinct variable and each leaf returns a distinct output value.  
For this $f$, we have $\DL(f) = 2^k$ and $\Ex[\ \DL(f{\uhr}\mb R_p)\ ] = (1+p)^k = \DL(f)^{\log_2(1+p)}$.
The same function also shows that $\gamma(p)$ in Corollary \ref{cor:orthogonal} cannot improved beyond $\log_2(1+p)$.
Since this function is not Boolean, it does not imply a lower bound on DNF shrinkage; however, a similar bound can be shown asymptotically by considering parity functions.

\section{Shrinkage of $\ACzero$ formulas}\label{sec:Linfty}\label{sec:ac0formulas}

Our bound the shrinkage DNF and CNF formulas implies an (only slightly weaker) bound on the shrinkage of depth-$2$ formula leaf-size.  We also discuss the relationship between leaf-size and a related size measure on $\ACzero$ formulas, the number of depth-$1$ gates.

\begin{df}
An {\em AC$^{\hspace{1pt}\text 0}$ formula} is a formula composed unbounded fan-in \AND{} and \OR{} gates with inputs labeled by literals.  We measure {\em depth} by the maximum number of gates on an input-to-output path; the expression ``depth-$d$ formula'' refers to an $\ACzero$ formula of depth at most $d$.
As with DeMorgan formulas, the {\em leaf-size} of an $\ACzero$ formula is the number of leaves labeled by literals.  An alternative size measure is the number of depth-$1$ gates (that have only literals as inputs). This number is at least half the total number of gates in any formula with no (useless) gates of fan-in $1$.

For a Boolean function $f$ and $d \ge 2$, we denote by $\mc L_d(f)$ the minimum leaf-size of depth-$d$ formula that computes $f$, and we denote by $\mc F_d(f)$ the minimum number of depth-$1$ gates in a depth-$d$ formula that computes $f$.  Note that $\mc L_d(f) = 1$ iff $f$ is a literal, and $\mc F_d(f) = 1$ iff $f$ is a nonempty conjunctive or disjunctive clause, and $\mc L_d(f) = \mc F_d(f) = 0$ iff $f$ is constant (hence computed by a single \AND{} or \OR{} gate with fan-in zero, which as a formula has no inputs and no depth-$1$ gates).

Finally, we denote by $\mc F(f)$ the minimum number of depth-$1$ gates in an (unbounded depth, unbounded fan-in) formula that computes $f$.
\end{df}

Note that $\mc F_2 = \min\{\DNF,\,\CNF\}$. Theorem \ref{thm:DLshrinkage} therefore implies:

\begin{cor}\label{cor:F2}
For all Boolean functions $f$,
\[
  \Ex[\ \mc F_2(f{\uhr}\mb R_p) + 1\ ] \le (\mc F_2(f)+1)^{\gamma(p)}.
\]
\end{cor}

Over $n$-variable Boolean functions, clearly $\mc F_d \le \mc L_d \le n\cdot\mc F_d$ and $\mc F \le \mc L \le n\cdot \mc F$.
The next lemma shows that, under a $1/2$-random restriction, $\mc F_d$ shrinks below $\mc L_d$ and $\mc F$ shrinks below $\mc L$ (independent of $n$).

\begin{la}\label{la:Linfty}
For all Boolean functions $f$ and $d \ge 2$, 
\[
\Ex[\ \mc L_d(f{\uhr}\mb R_{1/2})\ ] \le \mc F_d(f)
\quad\text{ and }\quad 
\Ex[\ \mc L(f{\uhr}\mb R_{1/2})\ ] \le \mc F(f).
\]
\end{la}

\begin{proof}
Let $F$ be a [depth-$d$] $\ACzero$ formula that computes $f$ using the minimum number of depth-$1$ gates.  By linearity of expectation, it suffices to show that each depth-$1$ subformula of $F$ (i.e.,\ conjunctive or disjunctive clause) has expected leaf-size at most $1$ under $\mb R_{1/2}$. 
Indeed, for any $k \ge 1$ and $p \in [0,1]$, 
\[
  \Ex[\ \mc L(\mr{\AND{}}_k{\uhr}\mb R_p) \ ] 
  =
  \Ex[\ \mc L(\mr{\OR{}}_k{\uhr}\mb R_p) \ ]
  &=
  \sum_{j=0}^k
  j \binom{k}{j} p^j \left(\frac{1-p}{2}\right)^{k-j}
  =
  kp\left(\frac{1-p}{2}\right)^{k-1}.
\]
When $p=\frac12$, we have $\frac{k}{2} \left(\frac34\right){}^{k-1} < 1$ for all $k \ge 1$.
\end{proof}

Using Lemma \ref{la:Linfty}, we obtain the following bound on the shrinkage of depth-2 formula leaf-size $\mc L_2$, which has a slightly worse exponent $\gamma(2p)$ compared to $\gamma(p)$ for $\mc F_2$ in Corollary \ref{cor:F2}.

\begin{cor}[Shrinkage of depth-2 formula leaf-size]\label{cor:depth2}
For all Boolean functions $f$,
\[
  \Ex[\ \mc L_2(f{\uhr}\mb R_p)+1\ ] \le (\mc L_2(f)+1)^{\gamma(2p)}.
\]
\end{cor}

\begin{proof}
Viewing $\mb R_p$ as a composition of $\mb R_{1/2}$ (first) and $\mb R_{2p}$ (second), we have
\[
  \Ex[\ \mc L_2(f{\uhr}\mb R_p)+1\ ] 
  &=
  \Ex_{\rho \sim \mb R_{2p}}\Big[\ 
    \Ex_{\sigma \sim \mb R_{1/2}}[\ 
      \mc L_2((f{\uhr}\rho){\uhr}\sigma)+1
    \ ]
  \ \Big]\\
  &\le
  \Ex_{\rho \sim \mb R_{2p}}[\ 
      \mc F_2(f{\uhr}\rho)+1
  \ ]
  &&\text{(Lemma \ref{la:Linfty})}\\
  &=
  (\mc F_2(f)+1)^{\gamma(2p)}
  &&\text{(Corollary \ref{cor:F2})}\\
  &\le
  (\mc L_2(f)+1)^{\gamma(2p)}
  &&\text{($\mc F_2 \le \mc L_2$)}.\vphantom{\big|}
  \qedhere
\]
\end{proof}

As an additional consequence of Lemma \ref{la:Linfty}, we observe that $\mc F$ has the same expected shrinkage factor (up to a constant factor) as DeMorgan leaf-size $\mc L$.

\begin{cor}[Shrinkage of unbounded fan-in, unbounded depth formulas]\label{cor:Linfty}
For all Boolean functions $f$,
\[
  \Ex[\ \mc F(f{\uhr}\mb R_p)\ ] 
  &= 
  O(\,p^2 \mc F(f) + p \sqrt{\mc F(f)}\,).
\]
\end{cor}

\begin{proof}
Assume $p \le 1/2$, since the bound is trivial otherwise. 
Viewing $\mb R_p$ as a composition of $\mb R_{2p}$ (first) and $\mb R_{1/2}$ (second), we have
\[
  \Ex[\ \mc F(f{\uhr}\mb R_p)\ ] 
  &=
  \Ex_{\sigma \sim \mb R_{1/2}}\Big[\ 
    \Ex_{\rho \sim \mb R_{2p}}[\ 
      \mc F((f{\uhr}\sigma){\uhr}\rho)
    \ ]
  \ \Big]\\
  &\le
  \Ex_{\sigma \sim \mb R_{1/2}}\Big[\ 
    \Ex_{\rho \sim \mb R_{2p}}[\ 
      \mc L((f{\uhr}\sigma){\uhr}\rho)
    \ ]
  \ \Big]
  &&\text{($\mc F \le \mc L$)}\\
  &=
  \Ex_{\sigma \sim \mb R_{1/2}}\left[\ 
    O\big(4p^2\mc L(f{\uhr}\sigma)
    + 2p\sqrt{\mc L(f{\uhr}\sigma)}\big)
  \ \right]
  &&\text{(Theorem \ref{thm:DeMorgan})}\\
  &=
  O\Big(p^2\Ex_{\sigma \sim \mb R_{1/2}}[\ \mc L(f{\uhr}\sigma)\ ]
  + p\sqrt{
  \smash{\Ex_{\sigma \sim \mb R_{1/2}}[\ \mc L(f{\uhr}\sigma)\ ]}
  \vphantom{|}}
  \,\Big)
  &&\text{(Jensen's inequality)}\\
  &=
  O(\,p^2\mc F(f)
  + p\sqrt{\mc F(f)}
  \,)
  &&\text{(Lemma \ref{la:Linfty})}.
  \qedhere
\]
\end{proof}

\section{Open problems}

We conclude by mentioning some questions raised by this work.

\begin{op}\label{op:bestgamma}
Determine the optimal function $\gamma_{\DL}(p)$ in Theorem \ref{thm:DLshrinkage}.  We have shown that 
\[
\ts\log_2(1+p) = \gamma_{\DT}(p) \le \gamma_{\DL}(p) \le \log_{\frac{2}{1-p}}(\frac{1+p}{1-p}).
\]
A simpler problem is to determine the least constant $C_{\DL}$ such that $\Ex[\ \DL(f{\uhr}\mb R_p)\ ] \le O(\DL(f)^{C_{\DL}\cdot p})$. It~follows from our bounds that $\frac{1}{\ln 2} = C_{\DT} \le C_{\DL} \le \frac{2}{\ln 2}$.
The same questions may be asked with respect to complexity measures $\ODL$, $\wODL$ and $\DNF$. 
\end{op}

\begin{op}\label{op:AC0}
Determine the shrinkage rate of depth-$d$ $\ACzero$ formulas for $d \ge 3$.  
We expect that
\begin{equation}\label{eq:conj}
  \Ex[\ \mc L_d(f{\uhr}\mb R_p)\ ]
  \le
  \mc L_d(f)^{O(p^{1/(d-1)})}.
\end{equation}
Ideally the constant in this big-$O$ should not depend on $d$.
\end{op}

We remark that inequality (\ref{eq:conj}) is known to hold for small $p = O(1/\log\mc L_d(f))^{d-1}$, when the bound is $O(1)$.  
This can be shown using the (Multi-)Switching Lemma of H{\aa}stad \cite{hastad2014correlation}.  It is also a direct consequence of the following result of the author \cite{rossman:LIPIcs:2019:10823}, which generalizes Corollary \ref{cor:SLsize} (on the decision tree size of decision lists) to $\ACzero$ formulas of any depth.

\begin{thm}[Decision tree size of $\ACzero$ formulas \cite{rossman:LIPIcs:2019:10823}]\label{thm:DTsizeAC0}
For all functions $f : \{0,1\}^n \to \{0,1\}$ computable by depth-$d$ $\ACzero$ formulas with fan-in $m$ (and leaf-size at most $nm^{d-1}$),
\[
  \Ex[\ \DT(f{\uhr}\mb R_p)\ ] \le 2
  \quad\text{ and }\quad
  \DT(f) \le O(2^{(1-p)n})
  \quad\text{ where }\quad p = O(1/\log m)^{d-1}.
\]
\end{thm}

A related question:

\begin{op}
Prove a stronger version of Theorem \ref{thm:DTsizeAC0} for depth-$d$ $\ACzero$ formulas with $m = \mc F_d(f)^{1/(d-1)}$ (instead of fan-in, which is larger for unbalanced formulas).  Such a result could be helpful in proving the shrinkage bound (\ref{eq:conj}).
\end{op}

Finally, we repeat the longstanding question concerning shrinkage of monotone formulas:

\begin{op}\label{op:m}
Determine the shrinkage exponent of monotone formulas. That is, find the maximum constant $\Gamma_{\mr{m}}$ such that 
\[
\Ex[\ \mc L_{\mr{m}}(f{\uhr}\mb R_p)\ ] \le O(p^{\Gamma_{\mr{m}}-o(1)} \mc L_{\mr{m}}(f) + 1)
\] 
for all monotone Boolean functions $f$, where $\mc L_{\mr{m}}$ is monotone formula leaf-size.  It is known that $2 = \Gamma_{\mr{DeMorgan}} \le \Gamma_{\mr{m}} \le \Gamma_{\mr{read\tu{-}once}} = \log_{\sqrt 5-1}(2) \approx 3.27$, and the second inequality is believed to be tight \cite{dubiner1993read,haastad1995shrinkage}.
\end{op}

\subsection*{Acknowledgements}

I am grateful to the anonymous referees of ITCS 2021 for their valuable comments and to the authors of \cite{lovett2020decision}, Shachar Lovett, Kewen Wu and Jiapeng Zhang, for stimulating conversations related to this work.  

\bibliographystyle{plain}
\bibliography{bens-allrefs.bib}

\end{document}